\documentclass[pra,superscriptaddress]{revtex4}
\usepackage[english]{babel}
\usepackage{graphicx}
\usepackage{dcolumn}
\usepackage{bm}
\usepackage{verbatim}
\usepackage{dsfont}
\usepackage{enumerate}
\usepackage{color}
\usepackage{mathrsfs}
\usepackage{braket}
\usepackage{slashed}

\usepackage{amssymb,amsmath,amsthm}
 
 \newtheorem{theorem}{Theorem}
\usepackage{pdflscape}           
\usepackage{epsfig}
\usepackage{multirow}
\begin{document}

\title{Tripartite Bell-type inequalities for measures of quantum coherence and skew information}
\author{Liang Qiu}
\affiliation{%
	Department of Physics, China University of Mining and Technology, Xuzhou, Jiangsu 221116, China%
	}
\email{lqiu@cumt.edu.cn}
\affiliation{%
    Institute for Quantum Science and Technology, University of Calgary, Alberta T2N 1N4, Canada%
	}

\author{Zhi Liu}
\affiliation{%
	Department of Physics, China University of Mining and Technology, Xuzhou, Jiangsu 221116, China%
	}

\author{Fei Pan}
\affiliation{%
	Department of Physics, China University of Mining and Technology, Xuzhou, Jiangsu 221116, China%
	}

\date{\today}
\begin{abstract}
We construct the tripartite Bell-type inequalities of product states for $l_{1}$-norm of coherence, relative entropy of coherence and skew information. Some three-qubit entangled states violate these inequalities. Particulary, the tripartite Bell-type inequality for relative entropy of coherence is always violated by the W class pure or mixed states as well as the GHZ class pure or mixed states, being used as entanglement witness.
\end{abstract}

\maketitle

\section{Introduction}
In a realistic theory, ¡°hidden¡± properties that the particles carry prior to measurement determine the measurement results, and the latter are independent of observation. In a local realistic theory, the measurements results obtained at two spacelike separation locations are independent of each other. The statistical correlations of measurements on multiparticle systems are imposed constraints by local realism. These constraints are in the form of Bell-type inequalities~ \cite{Bell1964,Greenberger1990,Clauser1969,Mermin1990,Ardehali1992,Belinski1993,Peres1999,Pitowsky2001,Weinfurter2001,Werner2001,Zukowski2002,
ZukowskiBrukner2002,Sen2002}, which was first proposed by Bell in Ref.~ \cite{Bell1964}. Quantum mechanics predicts violation of these inequalities, and it is known as Bell's theorem~\cite{Bell1964}. The physical reason of this phenomenon is that quantum nonlocality is a basic property of quantum states, and refers to the correlations that cannot be described by any local hidden variable theory~\cite{Bergmann2000}. If a composite state of a multipartite system cannot be written as a tensor product of states of the individual subsystems, the composite state is entangled. Particularly, nonlocality is the most important characteristic of quantum entanglement~\cite{Bell1964}. A pure bipartite entangled state always violates a Bell inequality, and thus, violation of Bell inequality forms necessary and sufficient criteria to detect entanglement in pure bipartite states~\cite{Gisin1991,Gisin1992}. Multipartite correlation function Bell inequalities, in which measurements on each particle can be chosen between two arbitrary dichotomic observables, have also been given~\cite{Weinfurter2001,Werner2001,Zukowski2002,ZukowskiBrukner2002,Sen2002}.

Quantum coherence has been established as an important notion, especially in the fields like quantum information theory~\cite{Baumgratz2014,Girolami2014,Streltsov2015,Monda2015}, quantum biology~\cite{Abbott2008,Plenio2008,Rebentrost2009,Lloyd2011,Huelga2013} and quantum thermodynamics~\cite{Rodriguez-Rosario2013,Lostaglio2015NC,Narasimhachar2015,Lostaglio2015PRX,Gardas2015}. In quantum information theory, quantum coherence is expected to be a resource~\cite{Baumgratz2014,Girolami2014,Streltsov2015,Monda2015}, and a resource theory of quantum coherence has also been put forward~\cite{Singh2015,Winter2016}. Particularly, Baumgratz et al. provided a quantitative theory of coherence as resource, and the quantum coherence measures such as the $l_{1}$-norm of coherence and the relative entropy of coherence were proposed in Ref.~\cite{Baumgratz2014}. Furthermore, the interplay between coherence and entanglement, quantum discord, mixedness has been given ~\cite{Streltsov2015,Singh2015PRA,Kumar2015,Xi2015,Yao2015,Cheng2015}. These results can be attributed to that resources are often interconvertible, which means trading one for another. An operational theory of coherence in quantum systems has been established by using the transformation processes like coherence distillation~\cite{Winter2016}. Very recently, Bu et al. constructed Bell-type inequalities of product states for $l_{1}$-norm of coherence and relative entropy of coherence in bipartite system~\cite{Bu2016}. In this paper, we push their results forward to the case of tripartite system, and establish the tripartite Bell-type inequalities for $l_{1}$-norm of coherence and relative entropy of coherence.

Skew information, motivated by the study of quantum measurements in the presence of a conserved quantity \cite{Araki1960,Yanase1961}, was firstly introduced by Wiger and Yanase to describe the information contents of mixed states in 1963 \cite{Wigner1963}. The statistical idea underlying skew information is shown to be the Fisher information in the theory of statistical estimation\cite{Luo2003}, and the Fisher information is not only a key notion of statistical inference \cite{Fisher1925,Cramer1974}, but also playing an increasing role in informational treats of physics \cite{Helstrom1976,Holevo1982,Frieden1995,Luo2002,Frieden2004}. By using skew information, an intrinsic measure for synthesizing quantum uncertainty of a mixed state has been proposed \cite{Luo2006}. The measure for correlations in terms of skew information is also given, and compared with the other measures of quantum correlations, the advantage of this measure is that its evaluation does not involve any optimization \cite{Luo2012}. Based on skew information, Girolami et al. defined and investigated a class of measures of bipartite quantum correlations of the discord type \cite{Modi2012} through local quantum uncertainty \cite{Girolami2013}. The skew information is also proven to be an asymmetry monotone in Ref. \cite{Girolami2014,Marvian2014}. Moreover, Girolami originally proposed the skew information as a coherence monotone \cite{Girolami2014}, however, it is now shown that such a quantity may increase under the action of incoherent operations \cite{Du2015,Sun2016}. In this paper, we establish the tripartite Bell-type inequalities for skew information motivated by the Bell-type inequality in quantum coherence.

The paper is organized as follows. In Sec. II, we introduce the definitions of the $l_{1}$-norm of coherence, the relative entropy of coherence, skew information and tripartite Bell inequality. Subsequently, we construct the tripartite Bell-type inequalities of three-qubit product states for two measures of coherence and skew information. In Sec. IV, we give examples that the tripartite Bell-type inequalities are violated by some typical three-qubit entangled states, and show that three-qubit W class pure or mixed states as well as GHZ pure or mixed states always violate the Bell-type inequality for relative measure of coherence. In the end, we conclude our paper with a discussion on our results in Sec. V.

\section{two measures of coherence, skew information and tripartite Bell inequality}
In this section, we introduce the definition of two measures of coherence, i.e., the $l_{1}$-norm of coherence and the relative entropy of coherence, skew information and the tripartite Mermin-Ardehali-Belinksii-Klyshko (MABK)~\cite{Mermin1990,Ardehali1992,Belinski1993} inequality.

\subsection{two measures of coherence}
Coherence is identified by the presence of off-diagonal terms in the density matrix. The basis of the axiomatic approach to quantify coherence consists of four postulates that any quantifier of coherence $\mathcal{C}$ should fulfill~\cite{Baumgratz2014}. Furthermore, the quantifiers of coherence depend on the reference basis, which means quantum states having different value of coherence in different basis. The $l_{1}$-norm of coherence and the relative entropy of coherence were first introduced and studied in Ref.~\cite{Baumgratz2014}, and they fulfil the four postulates. In the reference basis $\{|i\rangle\}$, the $l_{1}$-norm of coherence of a state $\rho$ is defined as
\begin{equation}
\mathcal{C}_{l_{1}}(\rho)=\sum\limits_{i,j\ i\neq j}|\rho_{ij}|=\sum\limits_{i,j}|\rho_{ij}|-1.
\end{equation}
The relative entropy of coherence is defined as
\begin{equation}
\mathcal{C}_{r}(\rho)=\min\limits_{\sigma\in \mathcal{I}}S(\rho\|\sigma)=S(\rho_{d})-S(\rho),
\end{equation}
where $\mathcal{I}$ is the set of all incoherent states in the reference basis $\{|i\rangle\}$, $S(\rho\|\sigma)={\rm Tr}\rho(\log_{2}\rho-\log_{2}\sigma)$ is the relative entropy between $\rho$ and $\sigma$, $S(\rho)=-{\rm Tr}\rho\log_{2}\rho$ is the von Neumann entropy of the state $\rho$, and $\rho_{d}$ is the diagonal matrix formed with diagonal elements of the state $\rho$ in the reference basis $\{|i\rangle\}$, i.e., $\rho_{d}=\sum_{i}\langle i|\rho|i\rangle|i\rangle\langle i|$. Here, ${\rm Tr}$ denotes the trace. For the canonical example of a maximally coherent state $|\Psi_{2}\rangle=\frac{1}{\sqrt{2}}\sum_{i=1}^{2}|i\rangle$ of $2$-dimensional system, $\mathcal{C}_{l_{1}}=1$ and $\mathcal{C}_{r}=1$.

\subsection{skew information}
The skew information is defined as
\begin{equation}
I(\rho,X)=-\frac{1}{2}{\rm Tr}[\rho^{1/2},X]^2,
\end{equation}
where $[\cdot,\cdot]$ denotes the commutator. $\rho$ is a general quantum state and $X$ is an observable, which is a self-adjoint operator and often a Hamiltonian served as a conserved quantity \cite{Luo2012}. The skew information provides an alternative measure of the information content for the quantum state $\rho$ with respect to observables not commuting with (i.e., skew to) the conserved quantity $X$ \cite{Wigner1963}. For the pure states, $I(\rho,X)$ reduces to the conventional variance $V(\rho,X):={\rm Tr}\rho X^2-({\rm Tr}\rho X)^2$. And generally, $0\leq I(\rho,X)\leq V(\rho,X)$. Particularly, $I(\rho,X)\leq1$ can be obtained for the state $\rho$ of the qubit system, and the equality can be saturated for pure qubit states~\cite{Mani2015}.

\subsection{Bell inequality}
For the case in which measurements on each particle can be chosen between two arbitrary dichotomic observables, \.{Z}ukowski et al. derived a single general Bell inequality, which they deemed a sufficient and necessary condition for the correlation function for $N$ particles to be describable in a local and realistic picture~\cite{Zukowski2002}. From this inequality, the Clauser-Horne-Shimony-Holt (CHSH) inequality~\cite{Clauser1969} for two-particle systems and the MABK inequalities for $N$ particles~\cite{Mermin1990,Ardehali1992,Belinski1993} can be obtained as corollaries. Specially, in the case of current interest, that of three qubits (labeled as $K=A,B$ or $C$), the inequality is
\begin{equation}
\left|M_{A}^{1}M_{B}^{1}M_{C}^{2}+M_{A}^{1}M_{B}^{2}M_{C}^{1}+M_{A}^{2}M_{B}^{1}M_{C}^{1}-M_{A}^{2}M_{B}^{2}M_{C}^2\right|\leq2,\label{eq4}
\end{equation}
where the superscripts $1$ and $2$ denote the two arbitrary dichotomic observables in which measurements on each particle can be chosen, and the measurement operators $M_{K}^{i}$ ($K=A,B,C$ and $i=1,2$) are the operators corresponding to measurements performing on qubit $K$.

\section{tripartite Bell-type inequality for product states}
In this section, we will construct tripartite Bell-type inequality for $l_{1}$-norm of coherence, relative entropy of coherence and skew information.

\subsection{tripartite Bell-type inequality for $l_{1}$-norm of coherence}
First of all, we build the tripartite Bell-type inequality of three-qubit product states for $l_{1}$-norm of coherence. In the following, if we choose observables $X,\ Y$ and $Z$ in $\mathcal{H}_{A}$, $\mathcal{H}_{B}$ and $\mathcal{H}_{C}$, respectively, the eigenvectors of $X,\ Y$ and $Z$ form the basis of $\mathcal{H}_{A}\otimes\mathcal{H}_{B}\otimes\mathcal{H}_{C}$. Therefore, $\mathcal{C}_{l_{1}}(\rho_{ABC},X,Y,Z)$ is used to denote the $l_{1}$-norm of coherence of the state $\rho_{ABC}$ in the basis formed by $X,\ Y$ and $Z$. Similarly, in the next subsection, the relative entropy of coherence of the state $\rho_{ABC}$ in the basis formed by $X,\ Y$ and $Z$ will be denoted by $\mathcal{C}_{r}(\rho_{ABC},X,Y,Z)$.
\begin{theorem}
For the product state $\rho_{ABC}=\rho_{A}\otimes\rho_{B}\otimes\rho_{C}$ in three-qubit system $\mathcal{H}_{A}\otimes \mathcal{H}_{B}\otimes \mathcal{H}_{C}$, and the observables $M_{K}^{i}$ ($K=A,B,C$ and $i=1,2$), there is
\begin{equation}
\begin{array}{ll}
&\mathcal{C}_{l_{1}}(\rho_{ABC},M_{A}^{1},M_{B}^{1},M_{C}^{2})+\mathcal{C}_{l_{1}}(\rho_{ABC},M_{A}^{1},M_{B}^{2},M_{C}^{1})\\
+&\mathcal{C}_{l_{1}}(\rho_{ABC},M_{A}^{2},M_{B}^{1},M_{C}^{1})-\mathcal{C}_{l_{1}}(\rho_{ABC},M_{A}^{2},M_{B}^{2},M_{C}^{2})\leq14.\label{eq5}
\end{array}
\end{equation}
\end{theorem}
\begin{proof}
For the product state $\rho_{A}\otimes\rho_{B}\otimes\rho_{C}$, it is straightforward to prove
\begin{equation}
\mathcal{C}_{l_{1}}(\rho_{ABC},X,Y,Z)+1=[\mathcal{C}_{l_{1}}(\rho_{A},X)+1][\mathcal{C}_{l_{1}}(\rho_{B},Y)+1][\mathcal{C}_{l_{1}}(\rho_{C},Z)+1].
\end{equation}
By using the above result, we have
\begin{equation}\begin{array}{ll}
&[\mathcal{C}_{l_{1}}(\rho_{ABC},M_{A}^{1},M_{B}^{1},M_{C}^{2})+1]+[\mathcal{C}_{l_{1}}(\rho_{ABC},M_{A}^{1},M_{B}^{2},M_{C}^{1})+1]\\
+&[\mathcal{C}_{l_{1}}(\rho_{ABC},M_{A}^{2},M_{B}^{1},M_{C}^{1})+1]-[\mathcal{C}_{l_{1}}(\rho_{ABC},M_{A}^{2},M_{B}^{2},M_{C}^{2})+1]\\
&=[\mathcal{C}_{l_{1}}(\rho_{A},M_{A}^{1})+1][\mathcal{C}_{l_{1}}(\rho_{B},M_{B}^{1})+1][\mathcal{C}_{l_{1}}(\rho_{C},M_{C}^{2})+1]\\
&+[\mathcal{C}_{l_{1}}(\rho_{A},M_{A}^{1})+1][\mathcal{C}_{l_{1}}(\rho_{B},M_{B}^{2})+1][\mathcal{C}_{l_{1}}(\rho_{C},M_{C}^{1})+1]\\
&+[\mathcal{C}_{l_{1}}(\rho_{A},M_{A}^{2})+1][\mathcal{C}_{l_{1}}(\rho_{B},M_{B}^{1})+1][\mathcal{C}_{l_{1}}(\rho_{C},M_{C}^{1})+1]\\
&-[\mathcal{C}_{l_{1}}(\rho_{A},M_{A}^{2})+1][\mathcal{C}_{l_{1}}(\rho_{B},M_{B}^{2})+1][\mathcal{C}_{l_{1}}(\rho_{C},M_{C}^{2})+1].\label{eq7}
\end{array}
\end{equation}
If we denote $\mathcal{C}_{l_{1}}(\rho_{K},M_{K}^{i})+1$ as $V_{K}^{i}$ ($K=A,B,C$ and $i=1,2$), the right-hand side of Eq.~(\ref{eq7}) can be simplified as
\begin{equation}
V_{A}^{1}V_{B}^{1}V_{C}^{2}+V_{A}^{1}V_{B}^{2}V_{C}^{1}+V_{A}^{2}V_{B}^{1}V_{C}^{1}-V_{A}^{2}V_{B}^{2}V_{C}^{2}.\label{eq8}
\end{equation}
Due to the fact that $1\leq V_{K}^{i}\leq 2$, the absolute value of the expression in Eq.~(\ref{eq8}) reads
\begin{equation}
\begin{array}{ll}
&|V_{A}^{1}V_{B}^{1}V_{C}^{2}+V_{A}^{1}V_{B}^{2}V_{C}^{1}+V_{A}^{2}V_{B}^{1}V_{C}^{1}-V_{A}^{2}V_{B}^{2}V_{C}^{2}|\\
=&|V_{A}^{1}(V_{B}^{1}V_{C}^{2}+V_{B}^{2}V_{C}^{1})+V_{A}^{2}(V_{B}^{1}V_{C}^{1}-V_{B}^{2}V_{C}^{2})|\\
\leq&|V_{A}^{1}(V_{B}^{1}V_{C}^{2}+V_{B}^{2}V_{C}^{1})|+|V_{A}^{2}(V_{B}^{1}V_{C}^{1}-V_{B}^{2}V_{C}^{2})|\\
=&2(|V_{B}^{1}V_{C}^{2}+V_{B}^{2}V_{C}^{1}|+|V_{B}^{1}V_{C}^{1}-V_{B}^{2}V_{C}^{2}|)\\
\leq& 16.
\end{array}
\end{equation}
Therefore, one can obtain
\begin{equation}\begin{array}{ll}
&[\mathcal{C}_{l_{1}}(\rho_{ABC},M_{A}^{1},M_{B}^{1},M_{C}^{2})+1]+[\mathcal{C}_{l_{1}}(\rho_{ABC},M_{A}^{1},M_{B}^{2},M_{C}^{1})+1]\\
+&[\mathcal{C}_{l_{1}}(\rho_{ABC},M_{A}^{2},M_{B}^{1},M_{C}^{1})+1]-[\mathcal{C}_{l_{1}}(\rho_{ABC},M_{A}^{2},M_{B}^{2},M_{C}^{2})+1]\\
\leq&16.
\end{array}
\end{equation}
From the above equation, it is found that the inequality in Eq.~(\ref{eq5}) exists. Hence, we complete the proof.
\end{proof}

The equality in Eq.~(\ref{eq5}) will be saturated when $\rho_{A},\ \rho_{B}$ and $\rho_{C}$ are all maximally coherent pure states, i.e., $\mathcal{C}_{l_{1}}(\rho_{K},M_{K}^{i})=1$ ($K=A,B,C$ and $i=1,2$). For the sake of simplicity, we denote the left-hand side of the inequality in Eq.~(\ref{eq5}) as $\mathcal{B}_{\mathcal{C}_{l_1}}$.

\subsection{tripartite Bell-type inequality for relative entropy of coherence}
In this subsection, we build the tripartite Bell-type inequality for relative entropy of coherence in analogy to that for $l_{1}$-norm of coherence.
\begin{theorem}
For the product state $\rho_{ABC}=\rho_{A}\otimes\rho_{B}\otimes\rho_{C}$ in three-qubit system $\mathcal{H}_{A}\otimes \mathcal{H}_{B}\otimes \mathcal{H}_{C}$, and the observables $M_{K}^{i}$ ($K=A,B,C$ and $i=1,2$), there is
\begin{equation}
\begin{array}{ll}
&\mathcal{C}_{r}(\rho_{ABC},M_{A}^{1},M_{B}^{1},M_{C}^{2})+\mathcal{C}_{r}(\rho_{ABC},M_{A}^{1},M_{B}^{2},M_{C}^{1})\\
+&\mathcal{C}_{r}(\rho_{ABC},M_{A}^{2},M_{B}^{1},M_{C}^{1})-\mathcal{C}_{r}(\rho_{ABC},M_{A}^{2},M_{B}^{2},M_{C}^{2})\leq6.\label{eq11}
\end{array}
\end{equation}
\end{theorem}
\begin{proof}
The relative entropy of coherence of the product state $\rho_{A}\otimes\rho_{B}\otimes\rho_{C}$ satisfies
\begin{equation}
\mathcal{C}_{r}(\rho_{A}\otimes\rho_{B}\otimes\rho_{C},X,Y,Z)=\mathcal{C}_{r}(\rho_{A},X)+\mathcal{C}_{r}(\rho_{B},Y)+\mathcal{C}_{r}(\rho_{C},Z).
\end{equation}
By using this result, the left-hand side of the Eq.~(\ref{eq11}) can be simplified as
\begin{equation}
2[\mathcal{C}_{r}(\rho_{A},M_{A}^{1})+\mathcal{C}_{r}(\rho_{B},M_{B}^{1})+\mathcal{C}_{r}(\rho_{C},M_{C}^{1})]\leq6.
\end{equation}
Here, we complete the proof.
\end{proof}
The equality in Eq.~(\ref{eq11}) is again saturated with $\rho_{A}$, $\rho_{B}$ and $\rho_{C}$ being the maximally coherent pure states. Henceforth, the left-hand side of Eq.~(\ref{eq11}) will be denoted as $\mathcal{B}_{\mathcal{C}_{r}}$ for the sake of convenience.

\subsection{tripartite Bell-type inequality for skew information}
After having established the tripartite Bell-type inequality for two measures of coherence, we will build the tripartite Bell-type inequality for skew information in this subsection. Here, if the observables $X$, $Y$ and $Z$ in $\mathcal{H}_{A}$, $\mathcal{H}_{B}$ and $\mathcal{H}_{C}$ are chosen, respectively, the skew information $I(\rho_{ABC},X,Y,Z)$ denotes $I(\rho_{ABC},X\otimes \mathbf{I}_{B}\otimes \mathbf{I}_{C}+\mathbf{I}_{A}\otimes Y\otimes\mathbf{I}_{C}+\mathbf{I}_{A}\otimes \mathbf{I}_{B}\otimes Z)$ for the sake of simplicity, where $\mathbf{I}_{K}\ (K=A,B,C)$ is the identity operator on system $H_{K}$.
\begin{theorem}
For the product state $\rho_{ABC}=\rho_{A}\otimes\rho_{B}\otimes\rho_{C}$ in three-qubit system $\mathcal{H}_{A}\otimes \mathcal{H}_{B}\otimes \mathcal{H}_{C}$, and the observables $M_{K}^{i}$ ($K=A,B,C$ and $i=1,2$), there is
\begin{equation}
\begin{array}{ll}
&I(\rho_{ABC},M_{A}^{1},M_{B}^{1},M_{C}^{2})+I(\rho_{ABC},M_{A}^{1},M_{B}^{2},M_{C}^{1})\\
+&I(\rho_{ABC},M_{A}^{2},M_{B}^{1},M_{C}^{1})-I(\rho_{ABC},M_{A}^{2},M_{B}^{2},M_{C}^{2})\leq6.\label{eq14}
\end{array}
\end{equation}
\end{theorem}
\begin{proof}
First of all, according to the definition of skew information, we have
\begin{equation}
\begin{array}{ll}
&I(\rho_{A}\otimes\rho_{B}\otimes\rho_{C},X,Y,Z)\\
=&-\frac{1}{2}{\rm Tr}[\sqrt{\rho_{A}}\otimes\sqrt{\rho_{B}}\otimes\sqrt{\rho_{C}},X\otimes \mathbf{I}_{B}\otimes \mathbf{I}_{C}+\mathbf{I}_{A}\otimes Y\otimes\mathbf{I}_{C}+\mathbf{I}_{A}\otimes \mathbf{I}_{B}\otimes Z]^2\\
=&-\frac{1}{2}{\rm Tr}\{[\sqrt{\rho_{A}},X]\otimes\sqrt{\rho_{B}}\otimes\sqrt{\rho_{C}}+\sqrt{\rho_{A}}\otimes[\sqrt{\rho_{B}},Y]\otimes\sqrt{\rho_{C}}
+\sqrt{\rho_{A}}\otimes\sqrt{\rho_{B}}\otimes[\sqrt{\rho_{C}},Z]\}^2\\
=&-\frac{1}{2}{\rm Tr}[\sqrt{\rho_{A}},X]^2-\frac{1}{2}{\rm Tr}[\sqrt{\rho_{B}},Y]^2-\frac{1}{2}{\rm Tr}[\sqrt{\rho_{C}},Z]^2\\
&-{\rm Tr}\{[\sqrt{\rho_{A}},X]\sqrt{\rho_{A}}\otimes\sqrt{\rho_{B}}[\sqrt{\rho_{B}},Y]\otimes\rho_{C}\}\\
&-{\rm Tr}\{[\sqrt{\rho_{A}},X]\sqrt{\rho_{A}}\otimes\rho_{B}\otimes\sqrt{\rho_{C}}[\sqrt{\rho_{C}},Z]\}\\
&-{\rm Tr}\{\rho_{A}\otimes[\sqrt{\rho_{B}},Y]\sqrt{\rho_{B}}\otimes\sqrt{\rho_{C}}[\sqrt{\rho_{C}},Z]\}\\
=&I(\rho_{A},X)+I(\rho_{B},Y)+I(\rho_{C},Z)-0-0-0\\
=&I(\rho_{A},X)+I(\rho_{B},Y)+I(\rho_{C},Z).
\end{array}
\end{equation}
By using the result, the left-hand side of Eq.~(\ref{eq14}) can be simplified as
\begin{equation}
2[I(\rho_{A},M_{A}^{1})+I(\rho_{B},M_{B}^{1})+I(\rho_{C},M_{C}^{1})]\leq6.
\end{equation}
Hence, we complete the proof.
\end{proof}

The equality in Eq.~(\ref{eq14}) can be saturated when $\rho_{A}$, $\rho_{B}$ and $\rho_{C}$ are the single qubit pure states~\cite{Mani2015}. In the following, the left-hand side of Eq.~(\ref{eq14}) will be denoted as $\mathcal{B}_{I}$ for the sake of convenience.

\section{violations of tripartite Bell-type inequality}
In the above section, we have constructed tripartite Bell-type inequalities for two measures of coherence and skew information. Now, we will give examples which violate these inequalities. Tripartite entangled states mean states with genuine entanglement of all three subsystems. There exist two subtypes of inequivalent states, i.e., GHZ state and W state~\cite{Dur2000}. Here, we just investigate the violation of the tripartite Bell-type inequality by the GHZ state and the W state, and the states obtained from them.

\textit{Example 1}. For two measures of coherence, the observables are assumed to be $M_{A}^{1}=\sigma_{x},\ M_{A}^{2}=\sigma_{z},\ M_{B}^{1}=-\sigma_{y},\ M_{B}^{2}=\sigma_{z},\ M_{C}^{1}=\sigma_{x},$ and $M_{C}^{2}=\sigma_{z}$. In other words, the basis formed by $M_{A}^{1},\ M_{A}^2$ for qubit $A$ are $\{(\frac{1}{\sqrt{2}},-\frac{1}{\sqrt{2}})^{T},(\frac{1}{\sqrt{2}},\frac{1}{\sqrt{2}})^{T}\},\ \{(1,0)^{T},(0,1)^{T}\}$. Similarly, the basis for qubit $B$ are $\{(\frac{1}{\sqrt{2}},\frac{i}{\sqrt{2}}
)^{T},(\frac{i}{\sqrt{2}},\frac{1}{\sqrt{2}})^{T}\},\ \{(1,0)^{T},(0,1)^{T}\}$, and the basis for qubit $C$ are $\{(\frac{1}{\sqrt{2}},-\frac{1}{\sqrt{2}})^{T},(\frac{1}{\sqrt{2}},\frac{1}{\sqrt{2}})^{T}\},\ \{(1,0)^{T},(0,1)^{T}\}$.
Now, we consider violations of the tripartite Bell-type inequalities for different states under these basis.

The W state $|{\rm W}\rangle=\frac{1}{\sqrt{3}}(|001\rangle+|010\rangle+|100\rangle)$. After some simple calculation, it is found $\mathcal{B}_{\mathcal{C}_{l_{1}}}\thickapprox15.876>14$, which indicates the W state violates the tripartite Bell-type inequality for $l_{1}$-norm of coherence given in Eq.~(\ref{eq5}). Similarly, $\mathcal{B}_{\mathcal{C}_{r}}\thickapprox6.503>6$ shows that the tripartite Bell-type inequality for relative entropy of coherence is violated by the W state.

The W class pure states $|\psi_{\rm W}\rangle=\cos{\theta}\cos{\phi}|001\rangle+\cos{\theta}\sin{\phi}|010\rangle+\sin{\theta}|100\rangle$, where $\theta\in(0,\pi)$ and $\phi\in(0,2\pi)$. After straightforward calculation, we can obtain the expressions of $\mathcal{B}_{\mathcal{C}_{l_{1}}}$ and $\mathcal{B}_{\mathcal{C}_{r}}$, however, they are analytically messy and therefore we have chosen to simply plot the exact, numerical results. In Fig.~\ref{fig1}, the plots of $\mathcal{B}_{\mathcal{C}_{l_{1}}}$ and $\mathcal{B}_{\mathcal{C}_{r}}$ as functions of the parameters $\theta$ and $\phi$ are given. From Fig.~\ref{fig1}(a), the value of $\mathcal{B}_{\mathcal{C}_{l_{1}}}$ is greater than $14$ in most of part given by $\theta$ and $\phi$, which means the violation of tripartite Bell-type inequality for $l_{1}$-norm of coherence by W class pure states. The Fig.~\ref{fig1}(b) shows that the value of $\mathcal{B}_{\mathcal{C}_{r}}$ is always greater than $6$, and the result indicates that the W class pure states always violate the tripartite Bell-type inequality for relative entropy of coherence. In other words, the latter can serve as an entanglement witness for the W class pure states.
\begin{figure}[h]
\begin{center}
\includegraphics[scale=1]{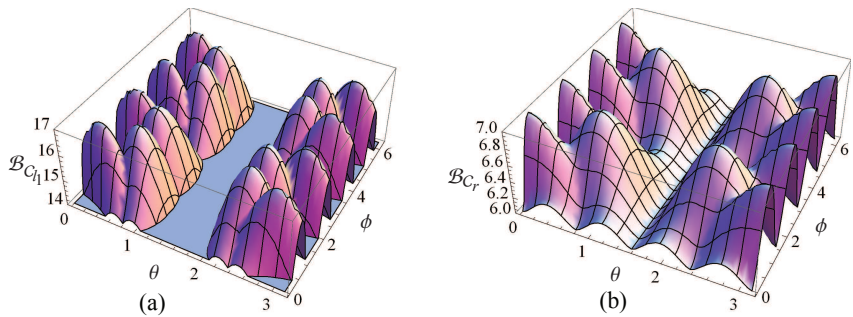}
\end{center}
\caption{(a) $\mathcal{B}_{\mathcal{C}_{l_{1}}}$, (b) $\mathcal{B}_{\mathcal{C}_{r}}$ as functions of the parameters $\theta$ and $\phi$ for the W class pure states.}
\label{fig1}
\end{figure}

The W class mixed states $\rho_{W}=\frac{p}{8}\mathbf{I}_{8}+(1-p)|W\rangle\langle W|$, where $\mathbf{I}_{8}$ is the identity matrix with dimensions $8$ and $p\in(0,1)$. These states are the generalized Werner states of three qubits according to the definition of Pittenger et al.~\cite{Pittenger2000}. With careful calculation, $\mathcal{B}_{\mathcal{C}_{l_{1}}}=\frac{1}{3}[31+16\sqrt{2}-(25+16\sqrt{2})p]$, which will be smaller than $14$ when $p\in((-17+16\sqrt{2})/(25+16\sqrt{2}),1)$. Therefore, the tripartite Bell-type inequality for $l_{1}$-norm of coherence will be violated by the W class mixed states if $p<(-17+16\sqrt{2})/(25+16\sqrt{2})\thickapprox0.118$. However, the case will be totally different for the tripartite Bell-type inequality for relative entropy of coherence. Now,
\begin{equation}
\mathcal{B}_{\mathcal{C}_{r}}=\frac{8-5p}{24}\log_{2}{\frac{8-5p}{24}}-\frac{4-p}{3}\log_{2}{\frac{4-p}{24}}+\frac{3p}{8}\log_{2}{\frac{p}{8}}
-\frac{2+p}{2}\log_{2}{\frac{2+p}{24}},
\end{equation}
and $\mathcal{B}_{\mathcal{C}_{r}}$ is greater than $6$ when $p\in(0,1)$. The result reveals that the W class mixed states always violate the tripartite Bell-type inequality for relative entropy of coherence, thus, the latter can be used as entanglement witness of the former.

The GHZ state $|{\rm GHZ}\rangle=\frac{1}{\sqrt{2}}(|000\rangle+|111\rangle)$. For this particular state, $\mathcal{B}_{\mathcal{C}_{l_{1}}}=20>14$ and $\mathcal{B}_{\mathcal{C}_{r}}=8>6$. Therefore, the GHZ state violates the tripartite Bell-type inequality in quantum coherence given in Eqs.~(\ref{eq5}) and ~(\ref{eq11}).

The GHZ class pure states $|\psi_{\rm GHZ}\rangle=\cos{\theta}|000\rangle+\sin{\theta}|111\rangle$, where $\theta\in(0,\pi)$. One can easily obtain $\mathcal{B}_{\mathcal{C}_{l_1}}=11+11|\sin{2\theta}|$. Thus, the tripartite Bell-type inequality for $l_{1}$-norm of coherence will be violated by the GHZ class pure states if $|\sin{2\theta}|>3/11$. In Fig.~\ref{fig2}(a), the evolution of $\mathcal{B}_{\mathcal{C}_{l_1}}$ versus the parameter $\theta$ has been numerically demonstrated, and the violation of the tripartite Bell-type inequality for $l_{1}$-norm of coherence occurs in most area of $\theta$. When it comes to the tripartite Bell-type inequality for relative entropy of coherence, $\mathcal{B}_{\mathcal{C}_{r}}=6-2(\cos^{2}{\theta}\log_{2}{\cos^{2}{\theta}}+\sin^{2}{\theta}\log_{2}{\sin^{2}{\theta}})$, which is always greater than or equal to $6$. That is, The GHZ class pure states violate the tripartite Bell-type inequality for relative entropy of coherence. This result can be intuitively observed from Fig.~\ref{fig2}(b), in which we plot the evolution of $\mathcal{B}_{\mathcal{C}_{l_1}}$ versus the parameter $\theta$.
\begin{figure}[h]
\begin{center}
\includegraphics[scale=1.35]{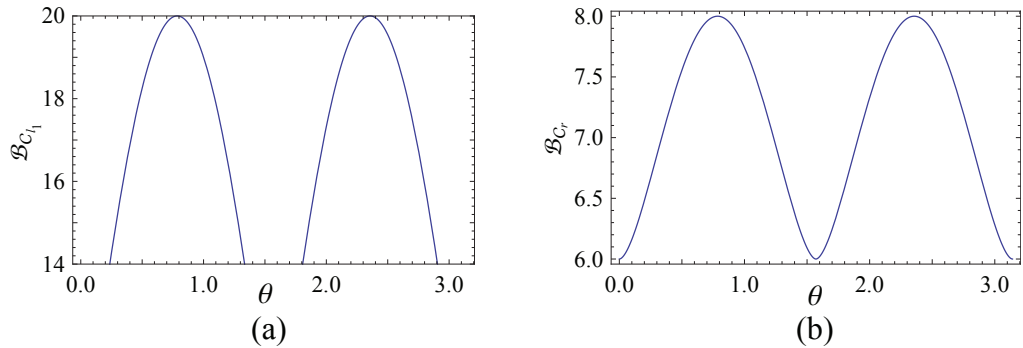}
\end{center}
\caption{(a) $\mathcal{B}_{\mathcal{C}_{l_{1}}}$, (b) $\mathcal{B}_{\mathcal{C}_{r}}$ as a function of the parameter $\theta$ for the GHZ class pure states.}
\label{fig2}
\end{figure}

The GHZ class mixed states $\rho_{\rm GHZ}=\frac{p}{8}\mathbf{I}_{8}+(1-p)|{\rm GHZ}\rangle\langle {\rm GHZ}|$. According to the definition given in Ref.~\cite{Pittenger2000}, these states are also the generalized three-qubit Werner states. After careful calculation, $\mathcal{B}_{\mathcal{C}_{l_{1}}}=22-20p$, from which it is easily to find the GHZ class mixed states cannot violate the tripartite Bell-type inequality for $l_{1}$-norm of coherence when $p\in[0.3,1)$. However, similar to the case of the W class mixed states, the tripartite Bell-type inequality for relative entropy of coherence is always violated by the GHZ-type mixed states, and now
\begin{equation}
\mathcal{B}_{\mathcal{C}_{r}}=6+\left(1-\frac{3p}{4}\right)\log_{2}{(4-3p)}+\frac{3p}{4}\log_{2}p.
\end{equation}
$\mathcal{B}_{\mathcal{C}_{r}}$ is greater than $6$ when $p\in(0,1)$. Therefore, the tripartite Bell-type inequality for relative entropy of coherence can serve as the entanglement witness for the GHZ class mixed states.

In Ref.~\cite{Bu2016}, the author claimed the bipartite Bell-inequality for relative entropy of coherence is violated by all two-qubit pure entangled states. Here, from the six cases given in the first example, the W class pure states, the W class mixed states, the GHZ class pure states and the GHZ class mixed states violate the tripartite Bell-type inequality for relative entropy of coherence. That is, different from the results in Ref.~\cite{Bu2016}, we present examples of the tripartite Bell-type inequality for relative entropy of coherence serving as entanglement witness for both pure states and mixed states.

In example 1, we have presented the entangled states and the reference basis such that the tripartite Bell-type inequalities in quantum coherence are violated. In the following, we give an example to indicate the tripartite Bell-type inequality for skew information can also be violated.

\textit{Example 2.} For the skew information, the observables are assumed to be $M_{A}^{1}=\sigma_{z},\ M_{A}^{2}=\sigma_{x},\ M_{B}^{1}=\cos{\frac{\pi}{6}}\sigma_{z}-\sin{\frac{\pi}{6}}\sigma_{x},\ M_{B}^{2}=\sin{\frac{\pi}{6}}\sigma_{z}+\cos{\frac{\pi}{6}}\sigma_{x},\ M_{C}^{1}=\cos{\frac{\pi}{3}}\sigma_{z}-\sin{\frac{\pi}{3}}\sigma_{x},$ and $M_{C}^{2}=\sin{\frac{\pi}{3}}\sigma_{z}+\cos{\frac{\pi}{3}}\sigma_{x}$.

Under these observables, we consider the violation of the tripartite Bell-type inequality for skew information.

The W state, the W class pure states and the W class mixed states. For the W state, $\mathcal{B}_{I}=10>6$ shows that the tripartite Bell-inequality is violated. The expression of $\mathcal{B}_{I}$ for the W class pure states is analytically messy, and we numerically plot $\mathcal{B}_{I}$ as functions of the parameters $\theta$ and $\phi$ in Fig.~\ref{fig3}(a). From the figure, the tripartite Bell-inequality for skew information is violated by the W class pure states in most of the area formed by the parameters $\theta$ and $\phi$. When it comes to the W class mixed states, $\mathcal{B}_{I}=10-7.5p-2.5\sqrt{8p-p^2}$, which will be greater than $6$ if $p\in\left(0,(11-\sqrt{57})/20\right)$.

The GHZ state, the GHZ class pure states and the GHZ class mixed states. For the GHZ state, we have $\mathcal{B}_{I}=13.464>6$, and the result shows that the GHZ state violates the tripartite Bell-inequality for skew information. For the GHZ class pure states, $\mathcal{B}_{I}=6+\sqrt{3}-(4+\sqrt{3})\cos{4\theta}$, and we plot $\mathcal{B}_{I}$ as a function of the parameter $\theta$ in Fig.~\ref{fig3}(b). It is obvious that $\cos{4\theta}<\sqrt{3}/(4+\sqrt{3})$ ensures the violation of the tripartite Bell-type inequality for skew information by the GHZ class pure states. In the end, $\mathcal{B}_{I}=(5+\sqrt{3})(4-3p-\sqrt{8p-7p^2})/2$ for the GHZ class mixed states. That is, when $p$ is greater than $(5-\sqrt{3})/11$, the tripartite Bell-type inequality for skew information cannot be violated by the GHZ class mixed states.
\begin{figure}[h]
\begin{center}
\includegraphics[scale=1.35]{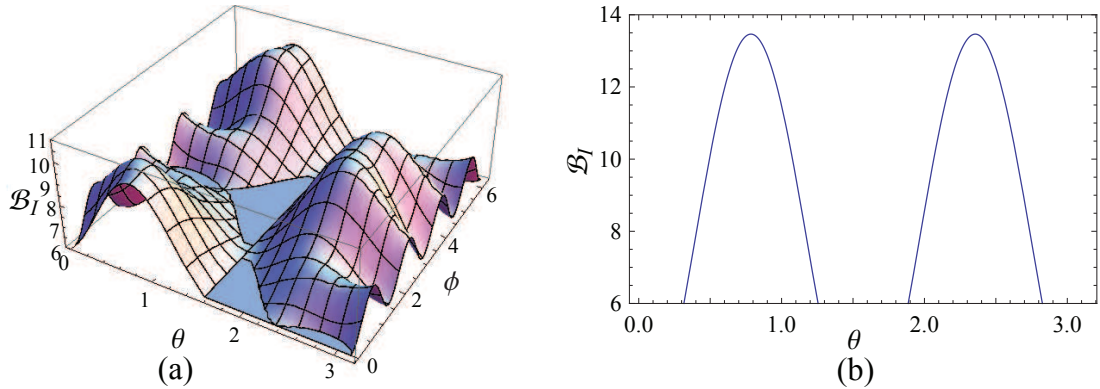}
\end{center}
\caption{The evolution of $\mathcal{B}_{I}$ versus (a) the parameters $\theta$ and $\phi$ for the W class pure states, (b) the parameter $\theta$ for the GHZ class pure states.}
\label{fig3}
\end{figure}

\section{Conclusions}
Bell's theorem draws distinction between quantum mechanics and classical mechanics, and it states that all of the predictions of quantum mechanics cannot be reproduced by physical theory of local hidden variables. Bell inequality, which could be tested experimentally, is a mathematical inequality derived from the locality and reality assumptions. In this work, we build the tripartite Bell-type inequalities for $l_1$-norm of coherence, relative entropy of coherence and skew information. These inequalities could be violated by the W state, the W class pure states, the W class mixed states, the GHZ state, the GHZ class pure states, and the GHZ class mixed states. Particularly, the W and GHZ class pure or mixed states always violate the tripartite inequality for relative entropy of coherence, and therefore the violation of the latter can serve as the entanglement witness for these entangled states.

\acknowledgements
L.Q. acknowledges the support from the Fundamental Research Funds for the Central Universities under Grant No. 2015QNA44.

\bibliography{bellinequality}

\end{document}